\documentclass{llncs}

\usepackage[table,pdftex,dvipsnames]{xcolor}
\usepackage[cmex10]{amsmath}
\usepackage{url}
\usepackage[caption=false,font=footnotesize]{subfig}
\usepackage{amsmath}
\usepackage{booktabs}
\usepackage{cite}
\usepackage[utf8]{inputenc}
\usepackage{stfloats}
\usepackage{amssymb}
\usepackage[ruled,linesnumbered,lined,boxed,commentsnumbered]{algorithm2e}
\usepackage{tikz}
\usepackage{multirow}
\usepackage{makecell}

\usetikzlibrary{shapes,arrows,positioning,calc}
\usepackage{tikz}

\usepackage{float}

\usepackage{xargs}                      % Use more than one optional parameter in a new commands

\newcommand{\xv}{{\boldsymbol {x}}}
\newcommand{\vv}{{\boldsymbol {v}}}
\newcommand{\va}{{\boldsymbol {a}}}
\newcommand{\vd}{{\boldsymbol {d}}}
\newcommand{\VM}{{\it VM}}
\newcommand{\CV}{{\it CV}}
\newcommand{\UB}{{\it UB}}
\newcommand{\M}{{\mathcal{M}}}

\usepackage[acronym]{glossaries}

\usepackage[colorinlistoftodos,prependcaption,textsize=tiny]{todonotes}
\newcommandx{\unsure}[2][1=]{\todo[linecolor=red,backgroundcolor=red!25,bordercolor=red,#1]{#2}}
\newcommandx{\change}[2][1=]{\todo[linecolor=blue,backgroundcolor=blue!25,bordercolor=blue,#1]{#2}}
\newcommandx{\info}[2][1=]{\todo[linecolor=OliveGreen,backgroundcolor=OliveGreen!25,bordercolor=OliveGreen,#1]{#2}}
\newcommandx{\improvement}[2][1=]{\todo[linecolor=Plum,backgroundcolor=Plum!25,bordercolor=Plum,#1]{#2}}
\newcommandx{\thiswillnotshow}[2][1=]{\todo[disable,#1]{#2}}
\DeclareMathOperator*{\argmin}{\arg\min}

\begin{document}

\title{Attacking the V:\\ 
On the Resiliency of Adaptive-Horizon MPC}

%% \title{LEREC: Level-Based Receding-Horizon Control for Attack-Resilient Formations}
%
\titlerunning{Attacking flight formations using global controller}  % abbreviated title (for running head)
%\author{}
\author{Scott A. Smolka\inst{1} \and Ashish Tiwari\inst{2} \and Lukas Esterle\inst{3} \and Anna Lukina\inst{3} \\ Junxing Yang\inst{1} \and Radu Grosu\inst{1,3}}
%
%% \authorrunning{omitted for review}%Lukina, Esterle, Hirsch, Bartocci, Yang, Tiwari, Smolka, Grosu} % abbreviated author list (for running head)
\tocauthor{}
%
%\institute{}
\institute{ Department of Computer Science, Stony Brook University, USA \\
    \and SRI International, USA \\
    \and Cyber-Physical Systems Group, Technische Universit\"at Wien, Austria
	}

\maketitle        % typeset the title of the contribution

\begin{abstract}
%% The abstract should summarize the contents of the paper
%% using at least 70 and at most 150 words. It will be set in 9-point
%% font size and be inset 1.0 cm from the right and left margins.
%% There will be two blank lines before and after the Abstract. \dots
We introduce the concept of a \emph{V-formation game} between a controller and an attacker, where controller's goal is to maneuver the plant (a simple model of flocking dynamics) into a V-formation, and the goal of the attacker is to prevent the controller from doing so.  
%% The controller can attain its goal by minimizing a certain flock-wide fitness 
%% function $J$, which is (almost) zero exactly when V-formation has been reached.  
%% Conversely, the attacker seeks to maximize $J$. 
%% We formalize V-formation games in terms of a Markov Decision Process (MDP) in which the controller and attacker jointly determine the transition probabilities.
Controllers in V-formation games utilize a new formulation of model-predictive control we call \emph{Adaptive-Horizon MPC} (AMPC), giving them extraordinary power: we prove that under certain controllability assumptions, an AMPC controller is able to attain V-formation with probability~1.
%% find a sequence of control actions (flock-wide accelerations) that brings the MDP to %% a V-formation goal state with probability one.

We define several classes of attackers, including those that in one move can remove
%%  a small number 
$R$ birds from the flock, or introduce random displacement
%% (perturbation)
into flock dynamics. 
%%again by selecting a small number of victim agents.  
We consider both \emph{naive attackers}, whose strategies are purely probabilistic, and \emph{AMPC-enabled attackers}, putting them on par strategically with the controllers.
%% in V-formation games.
While an AMPC-enabled controller is expected to win every game with probability~1, in practice, it is \emph{resource-constrained}: its maximum prediction horizon and the maximum number of game execution steps are fixed.  Under these conditions, an attacker has a much better chance of winning a V-formation game.

Our extensive performance evaluation of V-formation games uses statistical model checking to estimate the probability
%% by which 
an attacker can thwart the controller.  Our results show that for the bird-removal game with $R\,{=}\,1$, the controller almost always wins (restores the flock to a V-formation). For $R\,{=}\,2$, the game outcome critically depends on which two birds are removed.
%% : as long as the two removed birds are not adjacent to one another, the controller 
%% wins (is resilient).  
For the displacement game, our results again demonstrate that an intelligent attacker, i.e.~one that uses AMPC in this case, significantly outperforms its naive counterpart
%% that simply and uniformly at random
that randomly executes its attack.
%% \keywords{computational geometry, graph theory, Hamilton cycles}
\end{abstract}

	\section{Introduction}
\label{sec:intro}

%% \emph{Cyber-Physical Systems} (CPSs), which integrate computation, networking, and 
%% physical processes, are tightly coupled with our physical world.  Consequently, a
%% cyber-attack on a CPS can have serious \emph{physical} consequences, in particular to our 
%% physical infrastructure.  Examples include a 2015 malware attack on Ukraine’s power 
%% grid~\cite{ukraine}; the 2009 Stuxnet worm, which led to the destruction of nearly one 
%% fifth of Iran's nuclear centrifuges~\cite{stuxnet}; and a 2016 demonstration of the 
%% remote hacking of a Tesla Model~S~\cite{tesla}.  In contrast, traditional cyber-attacks 
%% have been more purely cyber-oriented, focusing for example on identity theft, financial 
%% gain, and the spread of propaganda.

Many Cyber-Physical Systems (CPSs) are highly distributed in nature, comprising a multitude of computing agents that can collectively exhibit \emph{emergent behavior}. A compelling example of such a distributed CPS is the \emph{drone swarm}, which are beginning to see increasing application in battlefield surveillance and reconnaisance~\cite{droneswarm}.  The emergent behavior they exhibit is that of \emph{flight formation}.

A particularly interesting form of flight formation is \emph{V-formation}, especially for long-range missions where energy conservation is key.  V-formation is emblematic of migratory birds such as Canada geese, where a bird flying in the \emph{upwash region} of the bird in front of it can enjoy significant energy savings.  The V-formation also offers a \emph{clear view} benefit, as no bird's field of vision is obstructed by another bird in the formation.  Because of the V-formation's intrinsic appeal, it is important to quantify the resiliency of the control algorithms underlying this class of multi-agent CPSs to various kinds of cyber-attacks.  This question provides the motivation for the investigation put forth in this paper.

%% From Ashish: In my opinion, here are the main messages of the paper:
%% 1. AMPC > MPC.  I.e. Adaptively changing length of horizon is a new and powerful idea to improve performance and resilience of MPC.
%% 2. Traditional feedback control is, by design,  resilient against noise, and consequently against certain kinds of attacks, but it may not 
%%    be resilient against smart attacks.  Adaptive control helps to guard against a larger class of attacks, but it can still falter due to 
%%    limited resources.
%% 3. Stochastic model checking is a promising approach to evaluate resilience of CPSs against classes of attacks.

\paragraph{Problem Statement and Summary of Results.}  We introduce the concept of \emph{V-formation games}, where the goal of the controller is to maneuver the plant (a simple model of flocking dynamics) into a V-formation, and the goal of the attacker is to prevent the controller from doing so.
%% The specific class of C-A games we consider is that of \emph{V-formation games}, where the 
%% controller can attain its goal by minimizing a certain flock-wide fitness function $J$, 
%% which is (almost) zero exactly when V-formation has been reached.  Conversely, the 
%% attacker seeks to maximize $J$.  We formalize V-formation games in terms of a Markov 
%% Decision Process (MDP) in which the controller and attacker jointly determine the 
%% transition probabilities.
Controllers in V-formation games utilize a new formulation of model-predictive control we call \emph{Adaptive Receding-Horizon MPC} (AMPC), giving them extraordinary power: we prove that under certain controllability conditions, an AMPC controller can attain V-formation with probability~1.

We define several classes of attackers, including those that in one move can remove a small number $R$ of birds from the flock, or introduce random displacement (perturbation) into the flock dynamics, again by selecting a small number of victim agents.  We consider both \emph{naive attackers}, whose strategies are purely probabilistic, and \emph{AMPC-enabled attackers}, putting them on par strategically with the controllers.  The architecture of a V-formation game with an AMPC-enabled attacker is shown in Figure~\ref{fig:ampc}.
While an AMPC-enabled controller is expected to win every game with probability~1, in practice, it is \emph{resource-constrained}: its maximum prediction horizon and the maximum number of game execution steps are fixed in advance.  Under these conditions, an attacker has a much better chance of winning a V-formation game.

AMPC is a key contribution of the work presented in this paper.  Traditional MPC uses a fixed finite \emph{prediction horizon} to determine the optimal control action.  Hence, it may get stuck in local minima. The AMPC procedure chooses it dynamically.  Thus, AMPC can adapt to the severity of the action played by its adversary by choosing its own horizon accordingly. The AMPC procedure is inspired by an adaptive optimization procedure recently presented in~\cite{lukina2016arxiv}.

Our extensive performance evaluation of V-formation games uses statistical model checking to estimate the probability that an attacker can thwart the controller.  Our results show that for the bird-removal game with $R\,{=}\,1$, the controller almost always wins (restores the flock to a V-formation). For $R\,{=}\,2$, the game outcome critically depends on which two birds are removed.
%% : as long as the two removed birds are not adjacent to one another, the controller 
%% wins (is resilient).  
For the displacement game, our results again demonstrate that an intelligent attacker, i.e.~one that uses AMPC in this case, significantly outperforms its naive counterpart
%% that simply and uniformly at random
that randomly executes its attack.

Traditional feedback control is, by design, resilient to noise, and also certain kinds of attacks; as our results show, however, it may not be resilient against smart attacks.  Adaptive-horizon control helps to guard against a larger class of attacks, but it can still falter due to limited resources.
Our results also demonstrate that statistical model checking represents a promising approach toward the evaluation CPS resilience against a wide range of attacks.

The rest of the paper is organized as follows.  Section~\ref{sec:background} introduces our dynamic model of V-formation in a flock of autonomous agents, and Section~\ref{sec:problem} defines our controller-attacker stochastic games. Section~\ref{sec:ampc} presents AMPC, and Section~\ref{sec:sgv} shows how AMPC is used in the V-formation games we consider. Section~\ref{sec:results} gives a critical analysis of our results, and Section~\ref{sec:related} discusses related work. Section~\ref{sec:conclusion} offers our concluding remarks and directions for future work.

    \section{V-Formation}
\label{sec:background}

We consider the problem of bringing a flock of $B$ birds from a random initial configuration to an organized V-formation.  Recently, Lukina et al.~\cite{lukina2016arxiv} have modeled this problem as a deterministic Markov Decision Process (MDP) $\mathcal{M}$, where the goal was to generate actions that caused $\mathcal{M}$ to reach a desired state. 

In our case, $\mathcal{M}$ is an MDP.  The \emph{state} of each bird in the flock is modeled using 4 variables: a 2-dimensional vector $\xv$ denoting the position of the bird in a 2D space, and a 2-dimensional vector $\vv$ denoting the velocity of the bird. Thus, the state space of $\M$ is $\mathbb{R}^{4B}$ representing a flock of $B$ birds.  The \emph{control actions} of each bird are 2-dimensional accelerations $\va$ and 2-dimensional position displacements $\vd$ (see discussion of $\va$ and $\vd$ below). Both are random variables.

Let $\xv_i(t),\vv_i(t),\va_i(t)$, and $\vd_i(t)$ denote the position, velocity, acceleration, and displacement of the $i$-th bird at time $t$, respectively.
%
%Here, each state of $\mathcal{M}$, at discrete time $t$, are of the form $(\xv_i(t),\vv_i(t))$, $1\,{\leqslant}\,i \,{\leqslant}\,N$, where $\xv_i(t)$ and $\vv_i(t)$ are vectors with $N$ elements (for an $N$-bird flock). Each element is 2-dimensional, representing positions and velocities, respectively. 
%
Then, the transition relation of the MDP $\mathcal{M}$ is given as follows:
\vspace*{-1mm}\begin{eqnarray}
\label{eq:trans}
 \xv_i(t + 1) &=& \xv_i(t) + \vv_i(t+1)\label{eq:x} + \vd_i(t) \qquad \forall~i\,{\in}\,\{1,\ldots,B\}, \nonumber\\
 \vv_i(t + 1) &=& \vv_i(t) + \va_i(t).\label{eq:v} %\\[-6mm]
\end{eqnarray}
Once the current acceleration and displacement are sampled, the next state is uniquely determined by~(\ref{eq:trans}) from the current state in $\M$~\cite{lukina2016arxiv}.

The problem of whether we can go from a random flock to a V-formation is a reachability question.  The reachability goal is the set of states representing a V-formation. A key assumption in~\cite{lukina2016arxiv} was that the reachability goal can be specified using a fitness function $J$,   
which assigns a non-negative real (fitness) value to each state in $\M$. 
%Moreover, $J$ was defined so that $J(s)\,{<}\,{\varphi}$, for a small $\varphi$, for exactly those states $s$ that correspond to a V-formation. {This is mentioned after eq.2} 
%Each state of $\mathcal{M}$ is represented by a fitness function reflecting essential motives for birds to fly in a V-formation in the first place.

The fitness of a state was determined by the following three terms: 
%% using three terms that measured if: (a)~birds have a clear visual field not blocked by any bird in front of them (\emph{clear view}); (b)~birds have the same velocity (\emph{velocity matching}); and
%% (c)~birds are located in space so that they are positioned in the upwash region of the bird(s) in front of them (\emph{upwash benefit}). 
\begin{itemize}
	\vspace*{-1.4mm}\item \emph{Clear View} ($\CV$). A bird's visual field is a cone with 
    angle $\theta$ that can be blocked by the wings of other birds. 
    The clear-view metric is defined by accumulating the percentage of a bird's visual 
    field that is blocked by other birds. The CV for the flock is the sum of the clear-view
    metric of all birds. The minimum value of $\CV$ 
    is $\CV^*{=}\,0$, and this value is attained in a perfect V-formation where all birds have clear view.
    
	\vspace*{1mm}\item \emph{Velocity Matching} ($\VM$). $\VM$ is defined as the  
	difference between the velocity of a given bird and all other birds, 
    summed up over all birds in the flock. The minimum
    value for $\VM$ is $\VM^*{=}\,0$, and this value is attained in a perfect
    V-formation where all birds have the same 
    velocity.
    
	\vspace*{1mm}\item \emph{Upwash Benefit} ($\UB$). The trailing upwash is 
    generated near the wingtips of a bird, while downwash is generated near 
    the center of a bird.  An upwash measure $um$ is defined on the 2D space using a 
    Gaussian-like model that peaks at the appropriate upwash and downwash regions. 
    %{Jesse: I commented this out as $um$ could be larger than 1 due to downwash} By design, the maximum $um$ for a bird is~1. 
    For bird~$i$ with upwash $um_i$, the upwash-benefit metric $\UB_i$ is defined as
    $1\,{-}um_i$, and $\UB$ for the flock is the sum of all $\UB_i$ $\forall~i\,{\in}\,\{1,\ldots,B\}$.
    The upwash benefit $\UB$ of a flock in V-formation is $\UB^*\,{=}\,1$, as all birds, except for the leader, have 
    minimum upwash-benefit metric ($\UB_i=0, um_i=1$),
    while the leader has upwash-benefit metric of $1$ ($\UB_i=1, um_i=0$).
\end{itemize}

\noindent{}Let $s=\{\xv_i, \vv_i\}_{i=1}^B$ be a state of a flock with $B$ birds. Given the above metrics, the overall fitness (cost) metric $J$ 
is of a sum-of-squares combination of $\VM$, $\CV$, and $\UB$ defined as follows:
\vspace*{-1mm}
\begin{align}
J(s) = (\CV(s)-\CV^*)^2 + (\VM(s)-\VM^*)^2 +(\UB(s)-\UB^*)^2.
\label{eq:fitness}
\end{align}
A state $s^{*}$ is considered to be a V-formation whenever $J(s^{*})\,{<}\,\varphi$, for a certain small threshold $\varphi$.

Given the above flocking model, the goal is to bring the flock from any configuration to a V-formation. Recall that we had two sets of control variables: accelerations $\vec{a}$ and displacements $\vec{d}$ for each bird of the flock.
We consider the scenario where the accelerations are under the control of one agent (the controller), and the displacements (position perturbations) are under the control of a second malicious agent (the attacker). This partition of the actions of the MDP into disjoint sets gives rise to a stochastic game on an MDP, which is described next.  
    \section{Controller-Attacker Games: Problem Definition}
\label{sec:problem}

%\todo[inline]{add a definition on what a game is! --> Russell Norvig!}
We are interested in games between a controller and an attacker, where the goal of the controller is to take the system to a desired set of states, and the goal of the attacker is to keep the system outside these states.  We formulate our problem by using a Markov Decision Processes (MDP) such that the controller and the attacker jointly determine the transition probabilities.

\begin{definition}A \textbf{Markov Decision Process (MDP)} $\M$ is a tuple $(S,A,T,J)$ consisting of: (1)~a set $S$ of states, (2)~a set $A$ of actions, (3)~a function $T: S\,{\times}\,A\,{\times}\,S\,{\mapsto}\,[0,1]$, where $T(s,a,s^\prime)$ is the probability of transitioning from state $s$ to state $s'$ under action $a$, and (4)~a function $J:S\,{\mapsto}\,\mathbb{R}$, where $J(s)$ is the reward (fitness) associated to state $s$.
%(with an initial state $s_0$), a set of actions $A$, a transition model $T$, and a cost function $J$. An MDP is called \textbf{deterministic} if $\forall s\in S,a\in A$ transition model $T$, such that $S\,{\times}\,{A}\,{\rightarrow}\,{S}$, specifies a unique state.
\end{definition}

\begin{figure}[t]
	\centering
		\begin{tikzpicture}[scale=1, every node/.style={transform shape},auto, node distance=2cm,>=latex', cross/.style={path 
		picture={ 
			\draw[black]
			(path picture bounding box.south east) -- (path picture bounding 
			box.north west) (path picture bounding box.south west) -- (path picture 
			bounding box.north east);
		}}]
		
		%%%%%%%%%%%%%%%% blocks
		\tikzstyle{arrow} = [thick,draw=black,->,>=stealth]
		\tikzstyle{proc} = [rectangle, rounded corners, minimum width=1cm, minimum 
		height=1cm,text centered, draw=black]
		\tikzstyle{output} = [coordinate]
		
		\node (start) [proc, text width=4.2cm] {\textbf{Controller}\\ \scriptsize{$\va(t) = {\it AMPC}\left(f, \xv(t), \vv(t), J\right)$}};
		\node (c1) [proc, text width=4.2cm] at (0,-1.3) {\textbf{Advanced Attacker}\\ \scriptsize{$\vd(t) = {\it AMPC}\left(g, \xv(t), \vv(t), -J\right)$}};
		\node (c2) [proc, text width=5.1cm, minimum height=2.3cm] at (6,-0.65) {\textbf{Flock}\\ \vspace*{0.1cm} \scriptsize{$\vv(t+1) = \vv(t) + \va(t)$}\\ \scriptsize{$\xv(t+1) = \xv(t) + \vv(t+1) + \vd(t)$}};
		\node (c3) [output,right=3.5em of start] {};
		\node (c4) [output,right=3.5em of c1] {};
		\node (c5) [output,right=1em of c2] {};
		\node (c6) [output,below=1em of c1] {};
		\node (c7) [output,left=1em of c1] {};
		\node (c8) [output,left=1em of start] {};
		
		%%%%%%%%%%%%%%%%%%% paths
		\draw[arrow] (start)-- node{\scriptsize{$\va(t)$}} (c3);
		\draw[arrow] (c1)-- node{\scriptsize{$\vd(t)$}} (c4);
		\draw[-,thick] (c2)--(c5) {};
		\draw[-,thick,near end] (c5)|- node{\scriptsize{$\xv(t+1),\vv(t+1)$}} (c6);
		\draw[-,thick] (c6)-|(c7) {};
		\draw[-,thick] (c7)-|(c8) {};
		\draw[arrow] (c7)--(c1) {};
		\draw[arrow] (c8)--(start) {};
		
\end{tikzpicture}
\vspace*{-2mm}
\caption{Controller-Attacker Game Architecture}
\label{fig:ampc}
     \vspace*{-3mm}
\end{figure}
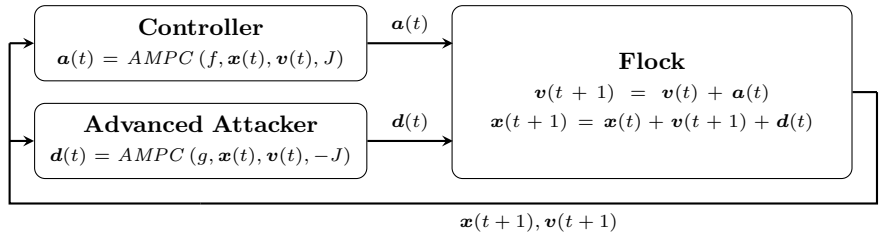

In a \textbf{stochastic game}~\cite{shapley1953stochastic}, the transition probability from state $s$ to state $s^\prime$ is controlled jointly by two players, a controller and an attacker in our case. To view an MDP as a stochastic game, we assume that the set of actions $A$ is given as a product $C\,{\times}\,D$, where the controller chooses the $C$-component of an action $\va$ and the attacker chooses the $D$-component of $\va$. We assume that the game is played in parallel by the controller and the attacker; i.e., they both take the state
$s(t)\in {S}$ of the system at time $t$,  
compute their respective actions $c(t)\in C$ and $d(t)\in D$, 
and then use the composed action $(c(t),d(t))$
to determine the
next state $s(t+1)\in {S}$ of the system (based on the transition function $T$).
% ASHISH: I commented out the text that was flock-specific (below), and replaced it by generic text above.
%$(\xv(t),\vv(t))\in{S}$ at time $t$, compute the accelerations $\va(t)$ and displacements $\vd(t)$ and pass them to the flock, which than computes the next state $(\xv(t+1),\vv(t+1))$.  

We consider {\em{randomized strategies}} for both the controller and the attacker. A randomized strategy is a mapping taking every state $s$ to a probability distribution $P(a\,{|}\,s)$ over the (available) actions. 
Once we fix a randomized strategy for the controller, and a randomized strategy for the attacker, the MDP reduces to a Markov chain on the state space $S$. Thus, 
the controller and attacker jointly fix the probability of transitioning from a state $s$ to a state $s^\prime$. 

In this paper, we consider {\em{reachability games}} only. In other words, we are given a set $G$ of ``good'' states and the goal of the controller is to reach a state in $G$. Let $s_0\,{\rightarrow}\,s_1\,{\rightarrow}\,s_2\,{\rightarrow}\,\cdots$ be a sequence of states (a run of the system). The controller wins on this run if $\exists{i}: s_i \in G$, and the attacker wins otherwise.

%\vspace*{-1.5mm}\begin{remark}
We are interested in discrete-time continuous-space dynamical systems.  Formally, the state space $S$ is $\mathbb{R}^n$ and the action space $A$ is in $\mathbb{R}^m$. In the bird flocking example, $n\,{=}\,m\,{=}\,4\,{\cdot}\,B$, where $B$ is the number of birds. We have four state variables and four action variables, respectively for each bird. They represent the $x$- and the $y$-components of the position $\xv_i$, velocity $\vv_i$, acceleration $\va_i$, and displacement $\vd_i$ of each bird $i$, respectively.
%\end{remark}

%\vspace*{-2mm}\begin{remark}
A classical problem in the study of games pertains to determining the existence of an optimal winning strategy (e.g.~a Nash equilibrium) for a player.  We are {\em{not}} concerned with such problems in this paper.  Due to the uncountably many states in the state- and action-space, solving such problems for our games of interest is extremely challenging. Instead, we focus on the problem of determining the likely winner of a game where the strategy of the two players is fixed.  Since we consider randomized strategies, determining the likely winner is a statistical model checking problem.  In other words, we want to evaluate the resilience of certain controllers under certain attack models.
%\end{remark}

%\begin{definition}
%First introduced by~\cite{shapley1953stochastic}, a \textbf{stochastic game} is a dynamic process proceeding from state (position) to state according to transition probabilities controlled jointly by two players.
%\end{definition}

We are now ready to formally define the problem we would like to solve.
\begin{definition}[Stochastic-game verification problem]
Let $\M\,{=}\,(S,A,T,J)$ be an MDP, where 
$A\,{=}\,C\,{\times}\,D$, and randomized strategies
$\sigma_C: S\,{\mapsto}\,PD(C)$ and
$\sigma_D: S\,{\mapsto}\,PD(D)$ mapping states $S$ to probability distributions
over $C$ and $D$. 
The \emph{stochastic-game verification problem} is to
determine the probability of reaching a state in $G\,{\subset}\,S$
in $m$ steps, for a given $m$, starting from an initial state (taken from a given probability distribution) in the underlying Markov chain induced by
strategies $\sigma_C,\sigma_D$ on the MDP $\M$.
\end{definition}

%\begin{remark}
Let us specify the randomized strategies. For a strategy
$\sigma$, we assume that we are
given a randomized algorithm that takes a state $s$ and
returns an action consistent with the probability distribution $\sigma(s)$.
%\end{remark}

%\begin{remark}
Our main interest here is in evaluating the resilience of a 
specific controller algorithm $\sigma_C$. The key assumption that
the controller and the attacker algorithms make is 
the existence of a {\emph{fitness function}}
$J: S\,{\mapsto}\,\mathbb{R}^{+}$ such that 
\begin{eqnarray*}
G & := & \left\{ s \mid J(s) < \varphi \mbox{ for some very small $\varphi > 0$} \right\}.
\end{eqnarray*}
Given such a fitness metric $J$, the controller works by minimizing the fitness of states reachable, in one or more steps, as it is done in model-predictive control (MPC).  Since the fitness function is highly nonlinear,
the controller uses an optimization procedure based on randomization to search for a minimum. Hence, our controller is a randomized procedure. One possible attack strategy we consider (for an advanced attacker) is based on the fitness function as well: the attacker tries to maximize the fitness of reachable states. 
%\end{remark}

A key contribution of our work is an adaptive MPC procedure called AMPC. Recall that traditional MPC uses a fixed finite horizon to determine the best control action. The AMPC procedure chooses it dynamically. Thus, AMPC can adapt to the severity of the action played by its adversary by choosing its own horizon accordingly. 
The AMPC procedure is inspired by an adaptive optimization procedure recently presented in~\cite{lukina2016arxiv}, which dynamically changes the amount of the effort it uses to search for a better solution in each step.  
The motivation for adaptation in~\cite{lukina2016arxiv} however was different, namely to take the optimizer out of a local minimum, and thus, ensure convergence to a global optimum.

    \section{The Adaptive-Horizon MPC Algorithm} % {Attack Strategies}
\label{sec:ampc}

We now present our new \emph{adaptive-horizon} \emph{model-predictive-control} algorithm, we call AMPC. We will use this algorithm as the controller strategy in the stochastic game on MDPs. We will also consider attack strategies that use AMPC. Since AMPC is an adaptive MPC procedure based on particle-swarm optimization (PSO), we first
briefly present background material on MPC and PSO.
%\todo[inline]{The rest of this section does not belong here. It should be moved in the description of the controller. Similar text can be used then for the advanced attacker.}

\subsection{Background on Model-Predictive Control}
%One way to find an action that can take a flock to a V-formation is based on using model-predictive control (MPC).  
Model-predictive control (MPC) determines the control action at current time $t$ by looking $h$ steps into the future and finding the best $h$-length sequence of control actions that can take the system from its current state $s(t)$ to a new state that has the lowest fitness. (Since we assume existence of a fitness metric $J$ that we are trying to minimize, we specialize the description of MPC to this case.) 
If ${s}_{\va^h}(t+h)$ denotes the state reached from state ${s}(t)$ in time $h$ following the actions $\va^h$ of length $h$, then in the MPC approach, at each time step $t$, the following minimization is performed to find the optimal set of actions
%\begin{align}
%&\textbf{opt-$\va$}^{h}(t)=\{\textbf{opt-$\va$}_i^{h}(t)\}_{i=1}^{b}=\argmin_{\va^h(t)}J(\boldsymbol{c}_{\va^h}(t+h)).
%\label{eq:opt}
%\end{align}
\begin{align}
&\textbf{opt-$\va$}^{h}(t)=\argmin_{\va^h(t)}J(s_{\va^h}(t+h)).
\label{eq:opt}
\end{align}
Since the model is an approximation of the system, only the first action $\va(t) = \textbf{opt-$\va$}^{1}(t)$ is applied as the action at time $t$,
and the remaining future $h-1$ actions found by the optimizer are ignored.  After the control action $\va(t)$ is applied, the system is left to evolve, and the process is repeated at $t\,{+}\,1, t\,{+}\,2,$ and so on.
%
%
%\vspace*{-1mm}
%\begin{align}
%J(\boldsymbol{c}(t),\va^h(t),{h}) = (\CV(\boldsymbol{c}_{\va}^{h}(t))-\CV^*)^2 &+ 
%(\VM(\boldsymbol{c}_{\va}^{h}(t))-\VM^*)^2 \nonumber \\ & +(\UB(\boldsymbol{c}_{\va}^{h}(t))-\UB^*)^2,
%\label{eq:fitness}
%\end{align}
%\noindent{}where ${h}$ is the length of the receding prediction horizon (RPH), $\va^h(t)$ is a sequence of accelerations of length ${h}$, and $\boldsymbol{c}_{\va}^{h}(t)$ is 
%the configuration reached after applying $\va^h(t)$ to $\boldsymbol{c}(t)$.
%
%

The MPC approach can be used for achieving a V-formation, as was outlined in~\cite{yang2016bda,yang2016love}.  These earlier works, however, did not use an adaptive dynamic window, and did not consider the adversarial control problem.
%We perform a single flock-wide minimization of $J$ at each time-step $t$ to obtain an optimal plan of  length $h$ of acceleration actions:\vspace*{-3mm}
%
%\begin{align}
%&\textbf{opt-$\va$}^{h}(t)=\{\textbf{opt-$\va$}_i^{h}(t)\}_{i=1}^{b}=\argmin_{\va^h(t)}J(\boldsymbol{c}(t),\va^h(t),{h}).
%\label{eq:opt}
%\end{align}
%\vspace*{-3mm}\noindent{} 
%
%We apply the first acceleration $\va_i(t) = \textbf{opt-$\va$}_i^{1}(t)$ as the optimal acceleration for bird $i$ %at time $t$. 
%

In MPC, optimization problem (\ref{eq:opt}) is additionally subject to constraints that bound the set of possible actions and states. For example, in our flocking model, the
magnitude of velocity and acceleration for each of the $B$ birds is bounded: $||\vv_i(t)||\,{\leqslant}\,\vv_{max},
||\va^h_i(t)||\,{\leqslant}\,\rho||\vv_i(t)||$ $\forall$ $i\,{\in}\,\{1,\ldots,B\}$,
where $\vv_{max}$ is a predefined constant and $\rho\,{\in}\,(0,1)$. 
% The initial state is selected following the given initial distribution. In the investigation of V-formation conducted in~\cite{lukina2016arxiv}, the initial positions and velocities of each bird are selected at random 
% within certain ranges, and limited such that the distance between any 
% two birds is greater than a (collision) constant $d_{min}$, and small
% enough so that each bird finds itself in the upwash or downwash region of at least one other bird. 

We use a \emph{particle-swarm-optimization algorithm} to solve the optimization problems generated by the MPC procedure.
%\todo[inline]{I added the following assuming we want to refer to the ARES approach - this might to be replaced if we want to use MPC instead. (I am not sure what Jesse used for his experiments right now)}
%\todo[inline]{Anna: as far as I understand Jesse is using MPC level-based with receding horizon, however, there are no PSO clones or important splitting}

% ASHISH: I think we should comment out the following; otherwise, we will look to be proposing something very close to TACAS paper.
%To solve the optimization problem, Lukina et al.~\cite{lukina2016arxiv} present an approach using \emph{Particle Swarm Optimization (PSO)}~\cite{Kennedy95particleswarm} to find potential next actions, combined with the idea of \emph{Importance Splitting}~\cite{kahn1951} in order to increase the chance of reaching a V-formation. Furthermore, they introduce adaptive horizons to allow for temporary worse situation in the flock enabling them to overcome local minima in the fitness function. Additionally, they introduced an adaptive number of particles for the PSO allowing them to better exploit the heuristics of PSO in combination with Importance Splitting as well as achieve a speed-up in the performance of their approach. However, their approach generates a plan offline to traverses the deterministic MDP in order to reach a specific state. In contrast, we present an approach to control the flock of birds online. The resilience of our approach is demonstrated in the presence of adversaries and noise.
 
    \subsection{Background on Particle Swarm Optimization}
\label{sec:swarmOptimization}

Particle Swarm Optimization (PSO) is a randomized approximation algorithm for determining the parameters that minimize a possibly nonlinear and possibly discontinuous cost (or fitness) function. PSO was first introduced by~\cite{Kennedy95particleswarm}. In an interesting twist of events, PSO took its original inspiration from bird flocking.

%As in~\cite{lukina2016arxiv}, our controller-PSO uses ``acceleration birds'' (these are the particles in the swarm). They should not be confused with the actual flocking birds. 
The PSO procedure is best described using the metaphor of a swarm of insects collaboratively trying to find the location of food. The insects, also called particles, live in the space defined by all possible valuations of the unknown parameters (of the optimization problem).  The food is located at the position where the objective function is minimized. PSO works by having a swarm of particles, which have the same goal of finding food (the reward) without knowing its location. Each particle is informed about its distance to the food (value of the objective function). The PSO algorithm repeatedly redistributes each particle towards the one closest to the food, with a speed proportional to the distance separating them, until all particles converge to the same position.

AMPC employs Matlab's toolbox $\texttt{particleswarm}$, which performs the classical version of PSO. A swarm of $p$ particles is sampled uniformly at random within a given bound on their positions and velocities. In the bird flocking example, if we try to find acceleration vectors by optimization over horizon $h$, then {\em{one}} ``particle'' represents $h$ 2-dimensional vectors for each of the $B$ birds, along with a vector of values that determine how these $h\cdot B$ acceleration vectors will be updated. 
%
%In the games we are considering each particle represents either a flock of bird-accelerations sequence $\{\va_i^{h}\}_{i=1}^b$, or a a flock of displacements sequence $\{\vd_i^{h}\}_{i=1}^b$, where $h$ is the current length of the receding horizon. 
After choosing a neighborhood of random size for each particle $j$, $j\,{\in}\,\{1,\ldots,p\}$, PSO computes the value of the given fitness function for each particle, and stores two vectors for each particle $j$: its so-far personal-best position $\mathbf{x}_{P}^j(t)$, and the position of its fittest neighbor $\mathbf{x}_{G}^j(t)$. The positions and velocities of the particle swarm $j\,{\in}\,\{1,\ldots,p\}$ are updated the following way:

\vspace*{-4mm}
\begin{align}
\mathbf{v}^j(t+1) = \omega\cdot\mathbf{v}^j(t) &+ y_1\cdot \mathbf{u_1}(t+1)\otimes(\mathbf{x}_{P}^j(t)-\mathbf{x}^j(t))  \nonumber \\
&+ y_2\cdot \mathbf{u_2}(t+1)\otimes(\mathbf{x}_{G}^j(t)-\mathbf{x}^j(t)),
\label{eq:swarm}
\end{align}

\vspace*{-1mm}\noindent{}where $\omega$ is an \emph{inertia weight}, which quantifies the trade-off between global and local exploration of the swarm (the value of $\omega$ is proportional to the exploration range); $y_1$ and $y_2$ are the \emph{self adjustment} and the \emph{social adjustment}, respectively; $\mathbf{u_1},\mathbf{u_2}\,{\in}\,{\rm Uniform}(0,1)$ are random variables; and $\otimes$ is the vector dot product, that is, $\forall$ random vector $\mathbf{z}$: $(\mathbf{z}_1,\ldots,\mathbf{z}_b)\otimes(\mathbf{x}_1^j,\ldots,\mathbf{x}_b^j)=(\mathbf{z}_1\mathbf{x}_1^j,\ldots,\mathbf{z}_b\mathbf{x}_b^j)$.

If the value of the fitness computed at each step of the PSO
for $\mathbf{x}^j(t+1)\,{=}\,\mathbf{x}^j(t)\,{+}\,\mathbf{v}^j(t+1)$ falls below the one for $\mathbf{x}_{P}^j(t)$, then $\mathbf{x}^j(t+1)$ is reassigned to $\mathbf{x}_{P}^j(t+1)$. A global best for the next iteration is determined as the particle with the best fitness among $j\,{\in}\,\{1,\ldots,p\}$. The stopping criterion of the PSO algorithm is either reaching the maximum number of iterations set in advance, or reaching the set time bound, or satisfying the minimum criterion. 

PSO can be used to solve any optimization problem. We use it to solve the optimization problem generated in the MPC approach. In a V-formation game, it can be used to obtain the birds' best accelerations, or even the best displacements, at each time step -- depending on whether MPC/PSO is being used by the controller or the attacker.

\emph{Remark}. We assume that PSO is fair, in the sense that it has a chance to sample all the points in the parameter space, and therefore it has the chance to find the optimal solution with probability one, given enough time.

%In a similar spirit, our advanced-attacker-PSO uses so called "displacement birds" (particles).

    	\subsection{The Main Algorithm of AMPC}
\label{sec:lerec}
% Inspired by the ARES algorithm in~\cite{lukina2016arxiv}, we propose a level-based receding-horizon model-predictive control algorithm we call LEREC. 
We propose the main algorithm of AMPC. %short for level-based Adaptive-horizon Model-Predictive Control {Jesse: this is already mentioned in the beginning of the section}. 
This algorithm performs step-by-step control of a given MDP $\M$ by looking $h$ steps ahead and predicting the next best state to move to.
We use PSO to identify the potentially best actions $\va^h$ in the current state achieving the optimal value of the fitness function in the next state.  For bird flocking, the fitness function, \texttt{Fitness}$(\M,\va^h,h)$ of $\va^h$ is defined as the minimum fitness metric $J$ obtained within $h$ steps by applying $\va^h$ on $\M$. Formally, we have
\vspace*{-2mm}
\begin{align}
\texttt{Fitness}(\M,\va^h,h) = \min_{1\leqslant \tau \leqslant h}{J(s_{\va^h}^\tau)}
\end{align}
where $s_{\va^h}^\tau$ is the state after apply the $\tau$th action of $\va^h$ on $\M$. For horizon $h$, PSO searches for the best sequence of 2-dimensional acceleration vector of length $h$, thus having $2Bh$ parameters to be optimized. The number of particles used in PSO is proportional to the number of parameters, i.e., $p = 2\beta B h$.
%was defined in equation~(\ref{eq:fitness}). 

The pseudocode for the AMPC algorithm is given in Algorithm~\ref{alg:lerec}. A novel feature of AMPC is that, unlike classical MPC that uses a fixed horizon $h$, AMPC adaptively chooses an $h$ depending on whether it is able to reach a fitness value that is lower than the current fitness by our chosen quanta $\Delta_i$, $\forall~i\,{\in}\,\{0,\ldots,m\}$.
%we would have had to resort to local optima without any guarantee of reaching a stable state.

%Following~\cite{lukina2016arxiv} we introduce level-based horizon as a way to overcome shortcomings of MPC and increase effectiveness of the optimization process. 
%To overcome the shortcoming of MPC and in order to increase the effectiveness of the optimization process, we introduce 
AMPC is hence an adaptive MPC procedure that uses level-based horizons.
It employs PSO to identify the potentially best next actions.
%for our flock.
If the chosen actions improve (decrease) the fitness of the next state $J(s_{k+h})$, $\forall~k\,{\in}\,\{0,\ldots,m\cdot h_{\mathit{max}}\}$, in comparison to the fitness of the previous state $J(s_k)$ by the predefined $\Delta_i$, the controller considers these actions to be worthy of leading the flock towards or keeping it in the V-formation.%
\footnote{We focus our attention on bird flocking, since the details generalize naturally to other MDPs that come with a fitness metric.}

In this case, the controller applies the actions to each bird and transitions to the next state of the MDP.
The threshold $\Delta_i$ determines the next level $\ell_i\,{=}\,J(s_{k+\widehat{h}})$ of the algorithm, where $\widehat{h} \leqslant h$ is the horizon with the best fitness. The prediction horizon $h$ is increased iteratively if the fitness has not been decreased enough. Upon reaching a new level, the horizon is reset to one (see Algorithm~\ref{alg:lerec}). Having a horizon $\widehat{h}\,{>}\,1$ means it will take multiple transitions in the MDP in order to reach a solution with improved fitness. However, when finding such a solution with $\widehat{h}\,{>}\,1$, we only apply the first action to transition the MDP to the next state. This is explained by the need to allow the other player (environment or an adversary) to apply their action before we obtain the actual next state. 
%adjustment faster and deal with unforeseen situations arising during runtime.
%
If no new level is reached within $h_{\mathit{max}}$ horizons, the first action of the best $\va^h$ using horizon $h_{\mathit{max}}$ is applied. 

The dynamic threshold $\Delta_i$ is defined as in~\cite{lukina2016arxiv}. Its initial value $\Delta_0$ is obtained by dividing the fitness range to be covered into $m$ equal parts, that is, $\Delta_0\,{=}\,(\ell_0\,{-}\,\ell_m)\,{/}\,m$, where $\ell_0\,{=}\,J(s_0)$ and $\ell_m\,{=}\,\varphi$. Subsequently, $\Delta_i$ is determined by the previously reached level $\ell_{i-1}$, as $\Delta_i\,{=}\,\ell_{i-1}{/}(m\,{-}\,i\,{+}\,1)$. This way AMPC advances only if $\ell_i\,{=}\,J(s_{k+\widehat{h}})$ is at least $\Delta_i$ apart from $\ell_{i-1}\,{=}\,J(s_{k})$.

This approach allows us to force PSO to 
escape from a local minimum, even if this implies passing over a bump, by gradually increasing the exploration horizon $h$. We assume that the MDP is controllable and that the set $G$ of good states is not empty, which means, that from any state, it is possible to reach a state whose fitness decreased by at least $\Delta_i$. 
%Figure~\ref{fig:approach} illustrates our approach.
Algorithm~\ref{alg:lerec} illustrates our approach.

%, which terminates if the stable state has been reached or the time elapsed.
%\begin{figure}[t]
%\centering
%	\input{diagram2}
%\vspace*{-2mm}
%\caption{Graphical representation of AMPC.}
%	\label{fig:approach}
    %\vspace*{-5mm}
%\end{figure}

%\input{algo_RPH}
%\SetAlgoSkip{}
%%%%%%%%% outer loop of level-based control
\begin{algorithm}[t]
	\SetKwFunction{particleswarm}{particleswarm}
    \SetKwFunction{Fitness}{Fitness}
    \SetKwFunction{RecedeHorizon}{RecedeHorizon}
	\SetKwInOut{Input}{Input}
	\SetKwInOut{Output}{Output}
    \DontPrintSemicolon
	\Input{$\M,\varphi,{h}_{\mathit{max}},m,B, \texttt{Fitness}$}
	\Output{$\{\va^i\}_{1\leqslant i\leqslant\,m}$ \textit{// optimal control sequence}}
	\BlankLine
	Initialize $\ell_0\leftarrow J(s_0)$; $\widehat{J}\leftarrow\inf$; ${p}\leftarrow 2\beta B h$; $i\leftarrow 1$;  ${h}\leftarrow 1$; $\Delta_0\leftarrow (\ell_0 - \varphi)/m$;
	\BlankLine
	\While{($\ell_{i-1} > \varphi)$ $\land$ $(i < m)$}
	{
     	\textit{// find and apply first best action out of the horizon sequence of length $h$}\;
        $[\va^{h},\widehat{J}]\leftarrow$\particleswarm{$\Fitness,\M,p,h$};\\

        \eIf{$\ell_{i-1}-\widehat{J} > \Delta_i \lor h = h_{\mathit{max}}$} 
		{
        \textit{// if a new level or the maximum horizon is reached}\;
        	$\va^i\leftarrow \va^{h}_1$;
            $\M\leftarrow\M^{\va^i}$;
            \textit{// apply the action and move to the next state}\;
			$\ell_i\leftarrow J(s(\M))$; \textit{// update $\ell_i$ with the fitness of the current state}\;
            $\Delta_i\leftarrow \ell_i/(m-i)$; \textit{// update the threshold on reaching the next level}\;
            $i \leftarrow i + 1$;
			${h} \leftarrow 1$;
			${p} \leftarrow 2\beta B h$; \textit{// update parameters}\;
		}
		{
			${h} \leftarrow {h} + 1$;
            $p\leftarrow 2\beta B h$; \textit{// increase the horizon}\;
		}
	}
	\caption{AMPC: Adaptive Model-Predictive Control}
	\label{alg:lerec}
\end{algorithm}
\setlength{\floatsep}{1cm}

\begin{theorem}[AMPC Convergence]
\label{thm:ampc}
Given an MDP $\M\,{=}\,(S,A,T,J)$ with positive and continuous fitness function $J$, and a nonempty set of target states $G\,{\subset}\,S$ with $G\,{=}\,\{s\,|\,J(s)\,{<}\,\varphi\}$. If the transition relation $T$ is controllable with actions in $A$, then there is a finite maximum horizon $h_{\mathit{max}}$ and a finite number of execution steps $m$, such that AMPC is able to find a sequence of actions $a_1,\ldots,a_m$ that brings a state in $S$ to a state in $G$ with probability one.
\end{theorem}

\begin{proof}
In each (macro) step of horizon length $h$, from level $\ell_{i-1}\,{=}\,J(s_k)$ to level $\ell_i\,{=}\,J(s_{k+\widehat{h}})$, AMPC decreases the distance to $\varphi$ by $\Delta_i\,{\geqslant}\,\Delta$, where $\Delta\,{>}\,0$ is fixed by the number of steps $m$ chosen in advance. Hence, AMPC converges to a state in $G$ in a finite number of steps, for a properly chosen $m$. AMPC is able to decrease the fitness in a macro step by $\Delta_i$ by the controllability assumption and the fairness assumption about the PSO algorithm. Since AMPC is a randomized algorithm, the result is probabilistic.
\end{proof}

Note that AMPC is a general procedure that performs adaptive MPC using PSO for dynamical systems that are controllable, come with a fitness metric, and have at least one optimal solution. In an adversarial situation two players have opposing objectives. The question arises what one player assumes about the other when computing its own action, which we discuss next.

    \newcommand{\anoise}{{\mathcal{AN}}}
\newcommand{\pnoise}{{\mathcal{PN}}}
\section{Stochastic Games for V-Formation}
\label{sec:sgv}

We describe the specialization of the stochastic-game verification problem to
V-formation.  In particular, we present the AMPC-based control strategy for reaching a V-formation, and the various attacker strategies against which we evaluate the resilience of our controller.

The MDP $\M$ for V-formation was presented in Section~\ref{sec:background}. The state variables of the MDP are the positions and velocities of the birds, and the control variables (defining the actions) are the accelerations and displacements. In the transition relation given in equation~(\ref{eq:v}), the attacker chooses the displacement $\vec{d}(t)$ it needs to manipulate the position of the birds,
whereas the controller chooses the acceleration $\vec{a}(t)$ to apply. Together, the pair $(\vec{a}(t),\vec{d}(t))$ defines the action that transforms one MDP state to another. We now define the controller's and attacker's strategies.

\subsection{Controller's Adaptive Strategies}

Given current state $(\vec{x}(t),\vec{v}(t))$, the controller's strategy $\sigma_C$ returns a probability distribution on the space of all possible accelerations (for all birds).  As mentioned above, this probability distribution is specified implicitly via a randomized algorithm that returns an actual acceleration (again for all birds).  This randomized algorithm is the AMPC algorithm, which inherits its randomization from the randomized PSO procedure it deploys.  

When the controller computes an acceleration, it assumes that the attacker does {\em{not}} introduce any disturbances; i.e., the controller uses the following model:
\vspace*{-4mm}\begin{eqnarray}
 \xv_i(t + 1) &=& \xv_i(t) + \vv_i(t+1) \qquad \forall~i\,{\in}\,\{1,\ldots,B\}, \nonumber \\
 \vv_i(t + 1) &=& \vv_i(t) + \va_i(t), \label{eq:noattack} %\\[-6mm]
\end{eqnarray}
where $\va(t)$ is the only control variable. Note that the controller chooses its next action $\va(t)$ based on the current configuration $(\xv(t),\vv(t))$ of the flock using MPC. The current configuration may have been influenced by the disturbance $\vec{d}(t-1)$ introduced by the attacker in the previous time step.  Hence, the current state need not be the state predicted by the controller when performing MPC in step $t-1$. Moreover, depending on the severity of the attacker action $\vec{d}(t-1)$, the AMPC procedure dynamically adapts its behavior, i.e.\ the choice of horizon $h$, in order to enable the controller to pick the best control action $\vec{a}(t)$ in response.

\subsection{Attacker's Strategies}

We are interested in evaluating the resilience of our V-formation controller when it is threatened by an attacker that can remove a certain number of birds from the flock, or manipulate a certain number of birds by taking control of their actuators (modeled by the displacement term in equation~(\ref{eq:trans})).
We assume that the attack lasts for a limited amount of time, after which the controller attempts to bring the system back into the good set of states. When there is no attack, the system behavior is the one given by equation~(\ref{eq:noattack}).

Note that there can be many different criteria for evaluating the success of an attack,  %(see Remark~\ref{remark:criteria})
but in our experiments, the controller is declared the winner if it can bring the flock to V-formation.
We consider three attack strategies (but see the future work discussion in Section~\ref{sec:conclusion}), each of which defines a V-formation game.

\vspace*{-0.5mm}\paragraph{\bf Remove Birds Game.}
In an RBG, the attacker selects a subset of $R$ birds, where $R\,{\ll}\,B$, and removes them from the flock.  The removal of bird $i$ from the flock at time $t\,{=}\,0$ can be simulated in our framework by allowing the attacker to set the displacement $\vd_i(0)$ for bird $i$ to $\infty$.  We assume that the flock is in a V-formation at time $t\,{=}\,0$.  
Thus, the goal of the controller is to bring the flock back into a V-formation consisting of $B\,{-}\,R$ birds.
%he controller needs to find the best adjustments in velocity $a_i$ for all remaining birds $i \in N - R$ during its turn. %$i \in N \wedge i \notin R$.
%Essentially, this results in a single-move game for the adversary. 
In an RBG, the attacker plays only one move.
When picking birds, the attacker is able to decide which birds will have the greatest negative impact on the flock's fitness when removed from the flock. Apart from seeing if the controller can bring the flock back to a V-formation, we also analyze the time it takes the controller to do so. 
%return to a v-formation for $R \leq \lceil\log(N)\rceil$ and 

% \todo[inline]{SAS: I would only suggest that the size R of the subset of
% birds removed from the flock (of size N) be such that R << N.  O/w I am
% not sure how interesting this game is.  Jesse has simulation results for
% R=1 and N=7.  Also, we should consider this game with and without process
% noise (PN), as Jesse has shown that the resiliency of the flock to remain
% in a V is highly dependent on the magnitude of PN.  It does very well with
% no PN or small PN, but resilience seems to degrade with increasing PN.}
%
%\begin{theorem}
%For any birds picked by the attacker, where $\left\vert{N - R}\right\vert \geq 3$, the planner can find 
%accelerations for each remaining bird in $N$ that will finally lead to a state $s^{*}$ such that cost 
%$J(s^{*})\{\leqslant}\,\varphi$.
%\end{theorem}

\vspace*{-0.5mm}\paragraph{\bf Random Displacement Game.}
In an RDG, the attacker chooses the displacement vector for a fixed number $R$ of birds uniformly from the space $[0,M]\times[0,2\pi]$. This means that the magnitude of the displacement vector is picked from the interval $[0,M]$, and the direction of the displacement vector is picked from the interval $[0,2\pi]$. We vary $M$ in our experiments. The $R$ birds that are picked in different steps are not necessarily the same, as the attacker makes this choice uniformly at random at runtime as well.
%In our second game, each player has control over all birds in the flock. The flock starts in a V-formation. However, both players have different goals and strategies. While the controller wants to keep the flock in a V-formation, the adversarial player tries to disrupt the V. Both players use the same planning approach but the controller tries to minimize the fitness function while the adversary tries to maximize the fitness in each step.
%In our second game, the adversarial player introduces malicious birds into the flock. These birds are controlled by the other player and hence can perturb the flock. To do so, the adversary adds small amounts of noise to this bird to distract the flock and disturb the v-formation. If this alone is not successful, the adversary can use a greater amount of noise to achieve the goal. However, this allows the controller to identify the adversary and henceforth ignore the malicious bird. 
The game starts from an initial V-formation. The attacker is allowed a fixed number of moves, say $20$, after which the displacement vector is identically $0$ for all birds.  The controller, which has been running in parallel with the attacker, is then tasked with moving the flock back to a V-formation, if necessary.
\vspace*{-0.5mm}\paragraph{\bf{AMPC Game.}}
An AMPC game is similar to an RDG except that the attacker does not use a uniform distribution to determine the displacement vector. The attacker is advanced and calculates the displacement (that will be the worst for the controller) using the AMPC procedure. See Figure~\ref{fig:ampc}.  In detail, the attacker applies AMPC, but assumes the controller applies zero acceleration. Thus, the attacker uses the following model of the flock dynamics:
\vspace*{-1mm}\begin{eqnarray}
 \xv_i(t + 1) &=& \xv_i(t) + \vv_i(t+1) + \vd_i(t) \qquad \forall~i\,{\in}\,\{1,\ldots,B\}, \nonumber \\
 \vv_i(t + 1) &=& \vv_i(t). \label{eq:attack} %\\[-6mm]
\end{eqnarray}
Note that the attacker is still allowed to have $\vd_i(t)$ be nonzero for a small number $R$ of birds. However, it can choose which $R$ birds it picks in each step.  It uses the AMPC procedure to simultaneously pick the $R$ birds and their displacements.

\begin{theorem}[AMPC resilience in a C-A game]
\label{thm:resilience}
Given a controller-attacker game, there is a finite maximum horizon $h_{\mathit{max}}$ and a finite maximum number of game-execution steps $m$ such that AMPC controller will win the controller-attacker game in $m$ steps with probability one.
\end{theorem}

\begin{proof}
Since the flock MDP (defined by Equation~6) is controllable, the PSO algorithm we use is fair, and the attack has a bounded duration, the proof of the theorem follows from Theorem~\ref{thm:ampc}. 
\end{proof}

\begin{remark}
While Theorem~\ref{thm:resilience} states that the controller is expected to win with probability one, we expect winning probability to be possibly lower than one in many cases because: (1)~the maximum horizon $h_{\mathit{max}}$ is fixed in advance, and so is (2) the maximum number of execution steps $m$; (3) the underlying PSO algorithm is also run with bounded number of particles and time.
\end{remark}

	\section{Statistical MC Evaluation of V-Formation Games}
\label{sec:results}
\newcommand{\majExp}{{2,000}}

As discussed in Section~\ref{sec:problem}, the
stochastic-game verification problem we address in the context of the V-formation-AMPC algorithm is formulated as follows.  Given a flock MDP $\M$ (we consider the case of $B\,{=}\,7$ birds), acceleration actions $\va$ of the controller, displacement actions $\vd$ of the attacker, the randomized strategy $\sigma_C: S\,{\mapsto}\,PD(C)$ of the controller (the AMPC algorithm), and a randomized strategy
$\sigma_D: S\,{\mapsto}\,PD(D)$ for the attacker,
determine the probability of reaching a state $s$ where the fitness function $J(s)\,{<}\,\varphi$ (V-formation in a 7-bird flock), starting from an initial state (in this case this is a V-formation), in the underlying Markov chain induced by strategies $\sigma_C$, $\sigma_D$ on $\M$.

Since the exact solution to this reachability is intractable due to the infinite/continuous space of states and actions, we solve it approximately with classical statistical model-checking (SMC).  The particular SMC procedure we use is described in~\cite{grosu2014isola} and based on an {\em additive} or {\em absolute-error $(\varepsilon,\delta)$-Monte-Carlo-approximation scheme}. This technique requires running $N$ i.i.d. game executions, each for a given maximum time horizon, computing if these executions reach a V-formation, and returning the average number of times this occurs.

The $N$ i.i.d. experiments determine the random variables $Z_1, ...,Z_N$, where the sample mean $\mu_Z\,{=}\,(Z_1\,{+}\,{\ldots}\,{+}\,Z_N)/N$ is assumed to be sufficiently greater than 0. In this case, one can exploit the Bernstein's inequality and fix $N$ to $\Upsilon\,{\propto}\,ln(1/\delta)/\varepsilon^2$. This results in an additive-error $(\varepsilon,\delta)$-approximation scheme:
\[
\Pr\left[\mu_Z\,{-}\,\varepsilon\leqslant\widetilde{\mu}_Z\leqslant\mu_Z\,{+}\,\varepsilon)\right]\geqslant{}1-\delta,
\]
where $\widetilde{\mu}_Z$ approximates $\mu_Z$ with absolute error $\varepsilon$ and probability $1-\delta$. In our case, each $Z_i$ is a Bernoulli random variable, where~1 means that the execution ends in a V-formation, and~0 means the opposite:
 
 \[Z=\left\{
\begin{array}{ll}
1, & \text{if}\:\:\exists t \in [0, m], J(s(t)) < \varphi,\\
0, & \text{otherwise}.
\end{array}\right.\]

This allows us to use the Chernoff-Hoeffding instantiation of the Bernstein's inequality, and fix the proportionality constant to $\Upsilon\,{=}\,4\,ln(2/\delta)/\varepsilon^2$, as in~\cite{HLMP04}. 

Each of the games described in Section~\ref{sec:sgv} is executed 2,000 times. For a confidence ratio $\delta\,{=}\,0.01$, we thus obtain an additive error of $\varepsilon\,{=}\, 0.1$.

We use the following parameters in the game executions: number of birds $B\,{=}\,7$, threshold on the fitness $\varphi\,{=}\,10^{-3}$, maximum horizon $h_{\mathit{max}}\,{=}\,5$, number of particles in PSO $p\,{=}\,20\,B\,h$. In RBG, the controller is allowed to run for a maximum of $30$ steps. In RDG and AMPC game, the attacker and the controller run in parallel for $20$ steps, after which the displacement becomes $0$, and the controller has a maximum of $20$ more steps to restore the flock to a V-formation.

To perform SMC evaluation of our AMPC approach we designed the above experiments in C and ran them on the Intel Core i7-5820K CPU with 3.30 GHz and with 32GB RAM available.\\
%\todo[inline]{We then present our experimental results and end with conclusions.}

\begin{figure}[t]
 \vspace*{-3ex}
 \centering
  \includegraphics[width=.495\textwidth]{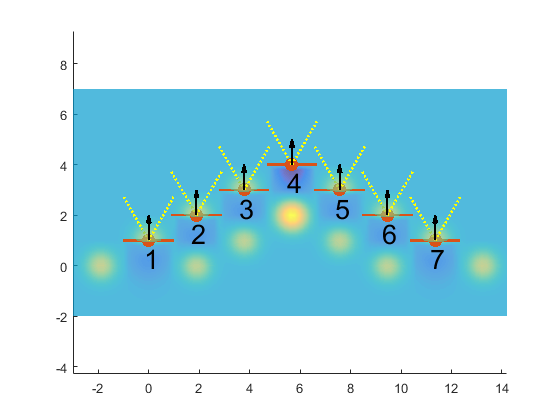}
  \includegraphics[width=.495\textwidth]{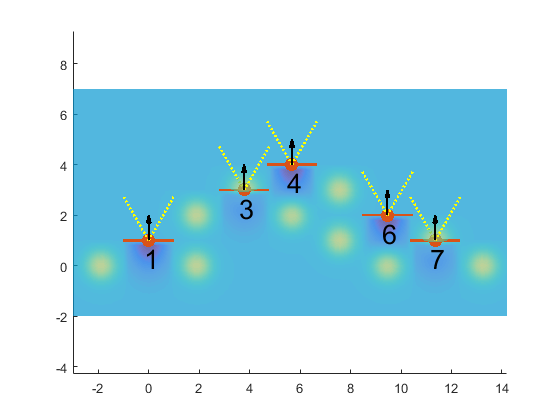}
  \vspace*{-5ex} 
  \caption{Left: numbering of the birds. Right: configuration after removing Bird 2 and 5. The red-filled circle and two protruding line segments represent a bird's body and wings. Arrows represent bird velocities. Dotted lines illustrate clear-view cones. A brighter/darker background color indicates a higher upwash/downwash.}
  \label{fig:numbering}
  \vspace{-3ex}
 \end{figure}

\begin{table}[h]
\centering
\vspace*{-4ex}
\caption {Results of 2,000 game executions for removing 1 bird with $h_{\mathit{max}}\,{=}\,5$,  $m\,{=}\,40$}
\vspace*{-1ex}
    \begin{tabular}{lcccccccccc}
    \toprule
		&&& Ctrl. success rate, \% &&& Avg. convergence duration &&& Avg. horizon\\ \midrule
    Bird 4  && & $99.9$           && & $12.75$  &&& $3.64$                            \\
    Bird 3  && & $99.8$           && & $18.98$  &&& $4.25$                            \\
    Bird 2  && & $100$           && & $10.82$  &&& $3.45$                            \\ \bottomrule
    \end{tabular}
\label{tab:resRemoveOne}
%\vspace*{-3ex}
\end{table}

%\vspace*{-4ex}
\begin{table}[ht]
%\vspace*{-3mm}
% 	\scriptsize
	\centering
	\caption{Results of 2,000 game executions for removing 2 birds with $h_{\mathit{max}} \,{=}\,5$, $m\,{=}\,30$}
	\begin{tabular}{lccccccccccc}
		\toprule
		&&& Ctrl. success rate, \% &&& Avg. convergence duration &&& Avg. horizon\\
      	\midrule
        Birds 2 and 3 & & &$0.8$ & & &$25.18$&&&$4.30$\\
        Birds 2 and 4 & & &$83.1$ & & &$11.11$ &&&$2.94$\\
        Birds 2 and 5 & & &$80.3$ & & &$9.59$ &&&$2.83$\\
        Birds 2 and 6 & & &$98.6$ & & &$7.02$ &&&$2.27$\\
        Birds 3 and 4 & & &$2.0$ & & &$22.86$ &&&$4.30$\\
        Birds 3 and 5 & & &$92.8$ & & &$11.8$ &&&$3.43$\\
		\bottomrule
\end{tabular}
\label{tab:resRemoveTwo}
%\vspace*{-3ex}
\end{table}

\begin{table}[t]
%\vspace*{-3mm}
% 	\scriptsize
	\centering
	\caption{Results of 2,000 game executions for random displacement and AMPC attacks with $h_{\mathit{max}}\,{=}\,5$ and $m\,{=}\,40$ (attacker runs for 20 steps)}
	\begin{tabular}{cccccccccccc}  
		\toprule
		Range of noise &&& Ctrl. success rate, \% &&& Avg. convergence duration &&& Avg. horizon\\
      	\midrule
        &&&\multicolumn{7}{c}{Random displacement game}\\
        \cmidrule(l){3-10}
        $[0,0.50]\times[0,2\pi]$ & & &$99.9$ & & &$3.33$ &&&$1.07$\\
        $[0,0.75]\times[0,2\pi]$ & & &$97.9$ & & &$3.61$ &&&$1.11$\\
        $[0,1.00]\times[0,2\pi]$ & & &$92.3$ & & &$4.14$ &&&$1.18$\\
 		\cmidrule(l){3-10}
        &&&\multicolumn{7}{c}{AMPC game}\\
        \cmidrule(l){3-10}
        $[0,0.50]\times[0,2\pi]$  && & $97.5$    &&& $4.29$ &&& $1.09$\\
        $[0,0.75]\times[0,2\pi]$  && & $63.4$    &&& $5.17$ &&& $1.23$\\
        $[0,1.00]\times[0,2\pi]$  && & $20.0$    &&& $7.30$ &&& $1.47$\\
		\bottomrule
\end{tabular}
\label{tab:resRandomNoise}
%\vspace*{-3ex}
\end{table}

\subsection{Discussion of the Results}
To demonstrate the resilience of our adaptive controller, for each game introduced in Section~\ref{sec:sgv}, we performed a number of experiments to estimate the probability of the controller winning.  Moreover, for the runs where the controller wins, the average number of steps required by the controller to bring the flock to a V-formation is reported as {\em{average convergence duration}}, and the average length of the horizon used by AMPC is reported as {\em{average horizon}}.

The numbering of the birds in Tables~\ref{tab:resRemoveOne} and~\ref{tab:resRemoveTwo} is given in Figure~\ref{fig:numbering}. Bird-removal scenarios that are symmetric with the ones in the tables are omitted. The results presented in Table~\ref{tab:resRemoveOne} are for the RBG game with $R\,{=}\,1$.
% The experiments clearly demonstrate the resilience of our adaptive controller. In the case where the attacker removes one bird, 
In this case, the controller is {\em{almost always}} able to bring the flock back to a V-formation, as is evident from Table~\ref{tab:resRemoveOne}. Note that removing Bird $1$ (or $7$) is a trivial case that results in a V-formation.
%The number of steps required to bring the flock back to a V-formation is reported as the {\em{convergence rate}}, and the average length of the horizon used by the AMPC is reported as {\em{average horizon}}.  

In the case when $R\,{=}\,2$, shown in Table~\ref{tab:resRemoveTwo}, the success rate of the controller depends on {\em{which two birds are removed}}. Naturally, there are cases where dropping two birds does not break the V-formation; for example, after dropping Birds~1 and~2, the remaining birds continue to be in a V-formation.  Such trivial cases are not shown in Table~\ref{tab:resRemoveTwo}. Note that the scenario of removing Bird~$1$ (or~$7$) and one other bird can be viewed as removing one bird in flock of $6$ birds, thus not considered in this table. Among the other nontrivial cases, the success rate of controller drops slightly in four cases, and drops drastically in remaining two cases. 
This suggests that attacker of a CPS system can incur more damage by being prudent in the choice of the attack. 
%when two birds are removed, but the adaptive controller is still able to reach a V-formation on $80$-$90$\% cases, as reported in Table~\ref{tab:resRemoveTwo}. 

Impressively, whenever the controller wins, the controller needs about the same number of steps to get back to V-formation (as in the one-bird removal case). On average, removal of two birds results in a configuration that has worse fitness compared to an RBG with $R\,{=}\,1$. %than a one-removed configuration. 
Hence, the adaptive controller is able to make bigger improvements (in each step) when challenged by worse configurations. Furthermore, among the four cases where the controller win rate is high, experimental results demonstrate that removing two birds positioned asymmetrically with respect to the leader poses a stronger, however, still manageable threat to the formation. For instance, the scenarios of removing birds~2 and~6 or 3 and~5 give the controller a significantly higher chance to recover from the attack, $98.6\%$ and $92.8\%$, respectively.

Table~\ref{tab:resRandomNoise} explores the effect of making the attacker smarter. Compared to an attacker that makes random changes in displacement, an attacker that uses AMPC to pick its action is able to win more often. This again shows that an attacker of a CPS system can improve its chances by cleverly choosing the attack. For example, the probability of success for the controller to recover drops from $92.3\%$ to $20.0\%$ when the attacker uses AMPC to pick displacements with magnitude in $[0,1]$ and direction in $[0,2\pi]$. The entries in the other two columns in Table~\ref{tab:resRandomNoise} reveal two even more interesting facts.

First, in the cases when the controller wins, we clearly see that the controller uses a longer look-ahead when facing a more challenging attack. This follows from the observation that the average horizon value increases with the strength of attack. This gives evidence for the fact that the adaptive component of our AMPC plays a pivotal role in providing resilience against sophisticated attacks.
Second, the average horizon still being in the range $1$-$1.5$, means that the adaptation in our AMPC procedure also helps it perform better than a fixed-horizon MPC procedure, where usually the horizon is fixed to $h\,{\geqslant}\,2$.
When a low value of $h$ (say $h\,{=}\,1$) suffices, the AMPC procedure avoids unnecessary calculation that using a fixed $h$ might incur.

In the cases where success rate was low (Row~5 in Table~\ref{tab:resRemoveTwo} and Row~6 in Table~\ref{tab:resRandomNoise}), we observed improved success rate ($9\%$ and $30.8\%$ respectively across $500$ runs) when we increased $h_{\mathit{max}}$ to~$10$ and $m$ to~$40$. This shows that success rate of AMPC improves as it is given more resources, as predicted by Theorem~\ref{thm:ampc}.

    \vspace*{-0.6ex}
\section{Related Work}
\label{sec:related}

In the field of CPS security, one of the most widely studied attacks is \emph{sensor spoofing}.  When sensors measurements are compromised, state estimation becomes challenging, which inspired a considerable amount of work on attack-resilient state estimation~\cite{DBLP:journals/tac/FawziTD14,Bullo13:TAC,Pajic14:ICCPS,Pajic15:ICCPS,UAVspoofing}.
In these approaches, resilience to attacks is typically achieved by assuming the presence of redundant sensors, or coding sensor outputs.
In our work, we do not consider sensor spoofing attacks, but assume the attacker gets control of the displacement vectors (for some of the birds/drones).  We have not explicitly stated the mechanism by which an attacker obtains this capability, but it is easy to envision ways (radio controller, attack via physical medium, or other channels~\cite{savage}) for doing so.

Adaptive control, and its special case of adaptive model predictive control, typically refers to the aspect of the controller updating its process model that it uses to compute the control action. The field of adaptive control is concerned with the discrepancy between the actual process and its model used by the controller. In our adaptive-horizon MPC, we adapt the lookahead horizon employed by the MPC, and not the model itself.  Hence, the work in this paper is orthogonal to what is done in adaptive control~\cite{adaptive_control,adaptive_mpc}.

A key focus in CPS security has also been detection of attacks. For example, recent work considers displacement-based attacks on formation flight~\cite{aiaa2016}, but it primarily concerned with detecting which UAV was attacked using an unknown-input-observer based approach. We are not concerned with detecting attacks, but establishing that the adaptive nature of our controller provides attack-resilience for free. Moreover, in our setting, for both the attacker the and controller the state of the plant is completely observable.

We are unaware of any work that uses statistical model checking to evaluate the resilience of adaptive controllers against (certain classes of) attacks.
    \vspace*{-0.6ex}
\section{Conclusions}
\label{sec:conclusion}

We have introduced AMPC, a new model-predictive controller that unlike MPC, comes with provable convergence guarantees. The key innovation of AMPC is that it dynamically adapts its receding horizon (RH) to get out of local minima. In each prediction step, AMPC calls
%% particle swarm optimizer
PSO with an optimal RH and corresponding number of particles. We used AMPC as a
bird-flocking controller whose goal is to achieve V-formation despite various forms of attacks, including bird-removal, bird-position-perturbation, and advanced AMPC-based attacks.  We quantified the resiliency of AMPC to such attacks using statistical model checking.  Our results show that AMPC is able to adapt to the severity of an attack by dynamically changing its horizon size and the number of particles used by PSO to completely recover from the attack, given a sufficiently long horizon and execution time (ET).  The intelligence of an attacker, however, makes a difference in the outcome of a game if RH and ET are bounded before the game begins.

Future work includes the consideration of additional forms of attacks, including: \emph{Energy attack}, when the flock is not traveling in a V-formation for a certain amount of time; \emph{Collisions}, when two birds are dangerously close to each other due to sensor spoofing or adversarial birds; and \emph{Heading change}, when the flock is diverted from its original destination (mission target) by a certain degree.

%2. adaptation helps to guarantee convergence, but here we observe that it can be used to obtain resilience to attacks.Effectively, if an optimization procedure can come out of local minima, it can be used o develop controllers that are resilient to attacks.
%3. stochastic model checking used to evaluate resilience of controllers to physical attacks on CPSs.
%4. attacker can employ sophisticated (adaptive) strategies (based on AMPC) to defeat controllers.
%5. Future work includes other atttacker strategies?  If so, which ones?  What else?
%Looking in the future, one of the most important for current CPS industry issue to address is attack detection and evaluation of the appropriate action plan, which we foresee to be possible to perform using our AMPC algorithm. This might require to consider distributed control of the birds. Moreover, in this work, due to the defined problem setting, we analyzed the resilience of our algorithm under the removal of one or two birds. In more general case with $B\,{>}\,7$ birds, we can use statistical approach to approximate the costs of removing $R\,{>}\,2$ birds form the V-formation.

%SECTIONS

\bibliographystyle{splncs03}
\bibliography{cavBib}

\end{document}